\documentclass[a4paper,11pt]{amsart}
\usepackage[dvips]{graphicx}
\usepackage{amssymb}
\usepackage{amsmath}
\usepackage[utf8]{inputenc}
\usepackage{amsthm}
\usepackage{amsfonts}
\usepackage{enumerate}
\usepackage{enumitem}
\usepackage[ruled,vlined]{algorithm2e}
\usepackage{hyperref}

\usepackage{young}
\usepackage[vcentermath]{youngtab}

\hypersetup{
  colorlinks   = true, 
  urlcolor     = blue, 
  linkcolor    = blue, 
  citecolor   = red 
}

\usepackage[a4paper]{geometry}
\usepackage{titletoc}
\usepackage{fancyhdr}
\newtheorem{thm}{Theorem} [section]
\newtheorem{prop}[thm]{Proposition}
\newtheorem{lem}[thm]{Lemma}
\newtheorem{coro}[thm]{Corollary} 

\theoremstyle{definition}
\newtheorem{defi}[thm]{Definition}

\theoremstyle{remark} 
\newtheorem{rem}[thm]{Remark}
\theoremstyle{definition}

\newtheorem*{nota*}{Notation}

\newcommand{\K}{{\mathbf K}}
\newcommand{\Q}{\overline{\mathbb Q}}
\newcommand{\C}{\mathbb{C}}
\newcommand{\Z}{\mathbb{Z}}
\newcommand{\N}{\mathbb{N}}
\renewcommand{\v}{{\boldsymbol v}}
\renewcommand{\r}{{\boldsymbol r}}

\newcommand{\e}{{\boldsymbol e}}

\newcommand{\f}{{\boldsymbol f}}
\newcommand{\g}{{\boldsymbol g}}
\newcommand{\h}{{\boldsymbol h}}
\newcommand{\x}{{\boldsymbol x}}
\newcommand{\y}{{\boldsymbol y}}
\newcommand{\puip}{\phi_p}
\newcommand{\puid}{\phi_d}
\newcommand{\val}{v_0}
\newcommand{\T}{\Theta}
\renewcommand{\t}{{\boldsymbol \theta}}
\newcommand{\Gtr}{\Psi} 
\newcommand{\gtr}{{\boldsymbol \psi}}
\newcommand{\Gam}{U}

\newcommand{\Ainv}{B}
\newcommand{\I}{\mathfrak{I}}

\newcommand{\V}{\mathfrak{V}}
\newcommand{\W}{\mathfrak{W}}

\newcommand{\inters}{\mathfrak{X}}
\newcommand{\intersd}{\mathfrak{X}_d}
\newcommand{\Const}{\Lambda}
\newcommand{\M}{{\rm M}}
\newcommand{\MM}{{\rm MM}}
\newcommand{\id}{I}

\numberwithin{equation}{section}

\begin{document}
\title{An algorithm to recognize regular singular Mahler systems}

\author{Colin Faverjon}
\address{Univ Lyon, Universit\'e Claude Bernard Lyon 1, CNRS UMR 5208, Institut Camille Jordan, F-69622 Villeurbanne, France}
\email{faverjon@math.univ-lyon1.fr}

\author{Marina Poulet}
\address{Univ Lyon, Universit\'e Claude Bernard Lyon 1, CNRS UMR 5208, Institut Camille Jordan, F-69622 Villeurbanne, France}
\email{poulet@math.univ-lyon1.fr}


\date{\today}

\thanks{The work of the first author was supported by the European Research Council (ERC) under the 
European Union's Horizon 2020 research and innovation programme 
under the Grant Agreement No 648132. The work of the second author was performed within the framework of the LABEX MILYON (ANR-10-LABX-0070) of Universit\'e de Lyon, within the program "Investissements d'Avenir" (ANR-11-IDEX- 0007) operated by the French National Research Agency (ANR)}

\begin{abstract}
This paper is devoted to the study of the analytic properties of Mahler systems at $0$. We give an effective characterisation of Mahler systems that are regular singular at $0$, that is, systems which are equivalent to constant ones. Similar characterisations already exist for differential and ($q$-)difference systems but they do not apply in the Mahler case. This work fills in the gap by giving an algorithm which decides whether or not a Mahler system is regular singular at $0$. In particular, it gives an effective characterisation of Mahler systems to which an analog of Schlesinger's density theorem applies.
\end{abstract}
\maketitle

\vspace{-1.5cm}
\setcounter{tocdepth}{1}
\setcounter{secnumdepth}{3}
\titlecontents{section}
[0.5em]
{}
{\contentslabel{1.3em}}
{\hspace*{-2.3em}}
{\titlerule*[1pc]{.}\contentspage}
{\addvspace{2em}\bfseries\large}
\titlecontents{subsection}
[3.8em]
{} 
{\contentslabel{2em}}
{\hspace*{-3.2em}}
{\titlerule*[1pc]{.}\contentspage}
\titlecontents{subsubsection}
[6.1em]
{} 
{\contentslabel{2.4em}}
{\hspace*{-4.1em}}
{\titlerule*[1pc]{.}\contentspage}
\renewcommand{\contentsname}{Contents}

 \pdfbookmark[0]{\contentsname}{Contents}
\tableofcontents

\vspace{-0.3cm}

\section{Introduction}

Let $\mathbf K$ be the field of Puiseux series with algebraic coefficients i.e. the field
$$
\mathbf K := \bigcup_{d \in \N^\star} \Q\left(\left(z^{1/d}\right)\right).
$$ 
For an integer $p\geq 2$, we define the operator
$$\begin{array}{rccl}
\puip :  & \K  & \rightarrow & \K  \\
         & f(z) & \mapsto & f\left(z^p\right).
\end{array}$$
The map $\puip$ naturally extends to matrices with entries in $\K$. A \emph{$p$-Mahler system} or, for short, a \emph{Mahler system} is a system of the form
\begin{equation}\label{eq:Mahler_at_0}
\puip\left(Y\right)=AY,\qquad A\in {\rm GL}_m\left(\Q\left(z\right)\right)\,.
\end{equation}
The study of Mahler systems began with the work of Mahler in 1929 \cite{Ma29,Ma30a,Ma30b}. Nowadays, there is an increased interest in their study because they are related to many areas such as automata theory or divide-and-conquer algorithms (see for example \cite{AF17,AF20,Co68,Du93,MF80,Ni97,Ph15} for a non-exhaustive bibliography).  In this paper, we focus on the singularity of Mahler systems at $0$. Similarly to differential or ($q$-)difference systems, a singularity at $0$ of a Mahler system can be regular. In that case, we say that the Mahler system is regular singular at $0$.

\begin{defi}
A $p$-Mahler system \eqref{eq:Mahler_at_0} is \emph{regular singular at $0$}, or for short \emph{regular singular}, if there exists a matrix $\Gtr\in {\rm GL}_m\left(\mathbf K\right)$ such that $\puip(\Gtr)^{-1}A\Gtr$ is a constant matrix.
\end{defi}

Note that some authors use the term ``Fuchsian'' to mean ``regular singular''. The singularities of differential systems and then of ($q$-)difference systems have been widely studied and algorithms have been given. 
One of the main interests in studying the regular singular systems is the good analytical properties of their solutions. 
A linear differential system is regular singular at $z=0$ if and only if all of its solutions have moderate growth at $z=0$, that is, at most a polynomial growth (see for example \cite[Th. 5.4]{VDPS03}). There exist criteria and algorithms to recognize regular singular differential systems, see for example \cite{Barkatou95,Bir13,Hi87,HW86,Moser59}. Then, algorithms have been given for other systems such as difference systems and $q$-difference systems (see for instance \cite{Barkatou89,BBP,BP96,Pr83}). In \cite{BBP}, the authors give a general algorithm for recognizing the regular singularity of linear functional equations satisfying some general properties. This algorithm applies to many systems such as differential systems and ($q$-)difference systems. However, this general algorithm does not apply to Mahler systems for the Mahler operator $\puip$ does not preserve the valuation at $0$. 
The aim of this paper is to fill this gap and to present an algorithm which decides whether or not a Mahler system is regular singular at $0$.

A remarkable property of linear differential systems at regular singular points is that Schlesinger's density theorem applies: the monodromy group is Zariski-dense in the Galois group of the system \cite[Cor. 5.2]{VDPS03}. An analog of this theorem was proved for ($q$-)difference systems also using the regular singular property (see \cite[Cor.~9.10]{vdPS97}, \cite[Th. 12.14]{vdPS97} or \cite{Sau03}). Recently, the second author has proved an analog of the Schlesinger's density theorem for Mahler systems \cite{Pou20} under regular singular conditions. Thus, the present work gives an algorithm to determine whether this theorem applies or not.

In general, a Mahler system does not admit a fundamental matrix of solutions in ${\rm GL}_m\left(\K\right)$. To find such a matrix, one has to consider some field extensions of $\K$. Let $\boldsymbol{\mathcal H}$ denote the field of Hahn series, that is the set of series of the form $\sum_{n \in \mathcal N} a_nz^n$, $\mathcal N\subset \mathbb Q$, where $\lbrace n\in\mathcal N \mid a_n\neq 0\rbrace$ is a well-ordered set (see \cite[Sec. 2]{Ro20}). One can extend the operator $\puip$ to $\boldsymbol{\mathcal H}$. In \cite{Ro20}, Roques proved that for every $p$-Mahler system there exists a matrix $\Gtr \in {\rm GL}_m(\boldsymbol{\mathcal H})$ such that $\puip(\Gtr)^{-1}A\Gtr$ is a constant matrix. Moreover, any constant Mahler system has a fundamental matrix of solutions in some ring containing some determinations of the functions $\log\log(z)$ and $\log^a(z)$, $a\in \Q\setminus\{0\}$ 
(see \cite{Ro18}). Thus, any Mahler system has a fundamental matrix of solutions of the form $\Gtr\Theta$, where $\Gtr$ is matrix with entries in $\boldsymbol{\mathcal H}$ and $\Theta$ is a fundamental matrix of solutions of a constant system. Among them, the regular singular systems are those for which one can choose $\Gtr$ in ${\rm GL}_m(\K)$. The restriction to the subfield $\K$ of $\boldsymbol{\mathcal H}$ is essential to preserve the analytic properties of the system. 
In particular, if $\f \in \K^m$ is a column vector, solution of a Mahler system, it follows from Rand\'e's Theorem \cite{BCR13, Ran92} that the entries of $\f$ are ramified meromorphic functions inside the unit disk.

\begin{defi}
Let $p\geq 2$ be an integer, and let $A, B \in {\rm GL}_m\left(\Q(z)\right)$. Let $\mathbf k \subset \boldsymbol{\mathcal H}$ be a field. The $p$-Mahler systems 
$$
\puip(Y)=AY\ \text{ and }\ \puip(Y)=BY
$$
are said to be $\mathbf k$-\emph{equivalent} if there exists a matrix $\Gtr \in {\rm GL}_m(\mathbf k)$ such that
$$
\puip(\Gtr)B=A\Gtr\, .
$$
In that case, the matrix $\Gtr$ is called an \emph{associated gauge transformation}.
\end{defi}

This choice of equivalence class ensures that if $Y$ is such that $\puip(Y)=AY$ then $\puip\left(\Gtr^{-1} Y\right)=B\left(\Gtr^{-1} Y\right)$. 
With this definition, the regular singular systems are the ones that are $\K$-equivalent to constant systems.

A Mahler system is said to be \emph{strictly Fuchsian at $0$} if the entries of $A$ are analytic functions at $0$ and $A(0)\in {\rm GL}_m\left(\Q\right)$. In other words, a system is strictly Fuchsian at $0$ if $0$ is not a singularity of this system.
It follows from \cite[Prop. 34]{Ro18} that systems which are strictly Fuchsian at $0$ are regular singular at $0$. Since the $p$-Mahler system associated with some matrix $A$ and the $p$-Mahler system associated with the matrix $z^\nu A$ for some $\nu\in \Z$ are ${\mathbf K}$-equivalent, there exist systems which are regular singular at $0$ but are not strictly Fuchsian at $0$.
Note that not all Mahler systems are regular singular at $0$. For example, the system 
$$
 \phi_2(Y)=\frac{1}{2}\left(\begin{array}{cc} 1 & 1 \\ \frac{1}{z}& \frac{-1}{z}  \end{array}\right)Y
$$
associated with the generating series of the \textit{Rudin-Shapiro sequence} is not regular singular at $0$ (see Section \ref{sec:Example}).
The main result of this paper reads as follows.

\begin{thm}\label{thm:algo}
Let $A \in {\rm GL}_m\left(\Q(z)\right)$ and  $p \geq 2$. There exists an algorithm which determines whether or not the Mahler system \eqref{eq:Mahler_at_0} is regular singular at $0$. This is done by computing the dimension of an explicit $\Q$-vector space.  If the system is regular singular at $0$ the algorithm computes a constant matrix to which the system is equivalent and a truncation at an arbitrary order of the Puiseux expansion of an associated gauge transformation.
\end{thm}

In \cite{CDDM18}, the authors built an algorithm to decide whether or not a linear homogeneous Mahler equation has a complete basis of solutions in $\K$. If it is the case, the associated Mahler system is regular singular at $0$ and $\K$-equivalent to the identity matrix. However, in general one does not know  \textit{a priori} the constant matrix to which a regular singular Mahler system is $\K$-equivalent. The algorithm mentioned in Theorem \ref{thm:algo} not only recognizes regular singular systems but also computes this constant matrix. From this point of view, the present work is a generalisation of this result of \cite{CDDM18}. From \cite{Ro20}, we know that every Mahler system is $\boldsymbol{\mathcal H}$-equivalent to a constant system but this work is not effective so it does not enable us to recognize regular singular Mahler systems. Our algorithm has been implemented in Python 3\footnote{The implemented algorithm is available at the following URL address:\\ \url{https://hal.archives-ouvertes.fr/hal-03147365/file/AlgoRegularSingularMahlerSyst.py}.}. Some bounds for the complexity are given in Section \ref{sec:Algo}.

\begin{rem}
It is also interesting to look at Mahler systems around other fixed points of $\puip$ such as $1$ or $\infty$. We say that a $p$-Mahler system is regular singular at $1$ (\textit{resp.} at $\infty$) if it is ${\bf L}$-equivalent to a constant matrix system, where we let ${\bf L}$ denote the field of Puiseux series $\bigcup_{d \in \N^\star} \Q\left(\left((z-1)^{1/d}\right)\right)$ (\textit{resp.} $\bigcup_{d \in \N^\star} \Q\left(\left(z^{-1/d}\right)\right)$). Using the change of variables $z=e^u$ one can test the regular singularity at $1$ using the theory of $q$-difference linear systems (with $q=p$). Furthermore, one can know if a system is regular singular at $\infty$ by applying Theorem \ref{thm:algo} to the system with matrix $A(1/z)$. 
\end{rem}

We present our strategy of proof. Assume that the Mahler system is $\mathbf K$-equivalent to a constant matrix $\Const$ with an associated gauge transformation $\Gtr \in {\rm GL}_m(\mathbf K)$. We can assume that $\Const$ is a Jordan matrix. Thus, using the relation $\puip(\Gtr)\Const=A\Gtr$, the columns $\gtr_1,\ldots,\gtr_m$ of $\Gtr$ are solutions of equations of the form
\begin{equation}\label{eq:Mahler_columns}
\lambda_i\puip(\gtr_i) + \epsilon_i \puip(\gtr_{i-1})= A\gtr_i
\end{equation}
where $\epsilon_i \in \{0,1\}$, $\lambda_i$ is an eigenvalue of $\Const$ and $\gtr_0:=0$. Thus, to prove that some Mahler system is regular singular at $0$, we will have to solve equations of the form \eqref{eq:Mahler_columns}. Section \ref{sec:Ramification_valuation} is devoted to the study of the solutions of such equations. We compute some bounds for their valuations and we find some admissible ramification indexes. Then we exhibit some linear equations which must be satisfied by the first coefficients of such solutions. In Section \ref{sec:preuvetheorem} we consider the vector space of solutions of such linear equations. Then we prove our main theorem (Theorem \ref{thm:Espace}) which states that the system is regular singular at $0$ if and only if the dimension of this vector space is precisely $m$. Then, in Section \ref{sec:Algo} we describe the algorithm of Theorem \ref{thm:algo} and we compute a bound for its complexity. Section \ref{sec:Example} is devoted to the study of some examples. Finally, in Section \ref{sec:OpenQuestions} we discuss some open problems.

\begin{nota*}
We let $\Q$ denote the algebraic closure of $\mathbb Q$ in $\C$ and $\Q^\star=\Q \setminus\{0\}$. We let $\val : \Q[[z]] \mapsto \Z \cap \{\infty\}$ denote the valuation at $z=0$: for $f\in\Q[[z]]$, $\val(f)$ is the supremum of the integers $v\in \N$ such that $f$ belongs to the ideal $z^v\Q[[z]]$. It extends uniquely as a valuation from $\K$ to $\mathbb Q$. We also extend it to the set of matrices with entries in $\K$ where $\val(U)$ denotes the minimum of the valuations at $0$ of the entries of a matrix $U$. 
Let $f$ be a polynomial. We let $\deg(f)$ denote the degree of $f$. If $M$ is a matrix with coefficients in $\Q(z)$, and $f$ is the least common multiple of the denominators of the entries of $M$, we define $\widetilde{M}:= fM$ and 
$$\deg(M):=\max\left(\deg(f),\deg\left(\widetilde{M}\right)\right).$$

Our bounds for the complexity of the algorithms presented here are given in terms of arithmetical operations in $\Q$. Given $f,g:\N \mapsto \mathbb R_{\geq 0}$ we use the classical Landau notation $f(n)=\mathcal O(g(n))$ if there exists a positive real number $\kappa$ such that $f(n) \leq \kappa g(n)$ for every large enough integer $n\in \N$. Similarly, we write $f(n)=\widetilde{\mathcal O}(g(n))$ if $f(n) =\mathcal O(g(n)\log(n)^c)$ for some $c\in \N$. Given a $n>0$, we let $\M(n)$ denote the complexity of the product of two polynomials of degree at most $n$, and $\MM(n)$ denote the complexity of the product of two matrices with at most $n$ rows and $n$ columns.

For the sake of clarity, we shall denote by roman capital letters $A,\Ainv,\ldots$ matrices whose coefficients are effectively known and by Greek capital letters $\Gtr, \T, \Const, \ldots$ the other matrices.
While matrices are denoted by capital letters $\Gtr,\T,\ldots$, the columns of these matrices should be denoted by bold lowercase letters $\gtr,\t,\ldots$
\end{nota*}

\section{Vector solutions of some Mahler equations} \label{sec:Ramification_valuation}

We fix some Mahler system \eqref{eq:Mahler_at_0}. Consider a vector of Puiseux series $\g \in \mathbf K^m$ and some nonzero algebraic number $\lambda \in \Q^\star$. The aim of this section is to study the solutions $\f\in \mathbf K^m$ of the system
\begin{equation}\label{eq:linear_system}
\lambda \puip(\f) + \puip(\g) = A\f\, .
\end{equation}
Precisely, we compute some integer $d$ such that any solution of this equation belongs to $\Q\left(\left(z^{1/d}\right)\right)^m$ assuming that $\g$ also belongs to $\Q\left(\left(z^{1/d}\right)\right)^m$. Then, this integer $d$ being fixed, we exhibit an integer $\nu_d$ such that $\val(\f)$, the minimum of the valuations of the entries of the vector $\f$, is at least $\nu_d/d$, assuming that $\val(\g) \geq \nu_d/d$. Finally, we prove that $\f$ is uniquely determined by its first coefficients and that these coefficients must satisfy some linear equations. 

\subsection{The cyclic vector lemma}
In \cite{CDDM18}, the authors developed a method to solve linear Mahler equations, that is, equations of the form
\begin{equation}\label{eq:eqMahler}
q_0 y + q_1\puip(y)+q_2 \puip^{2}(y)+\cdots+q_{m-1}\puip^{m-1}(y)-\puip^{m}(y)=0\, ,
\end{equation}
with $q_0,\ldots,q_{m-1} \in \Q(z)$. Actually, one can use these results to solve linear systems of the form \eqref{eq:linear_system}. In order to do that,
we use a result known as the cyclic vector lemma. For the sake of completeness, we develop here a proof of this result.

\begin{thm}[cyclic vector lemma]\label{thm:CyclicVector}
Any Mahler system \eqref{eq:Mahler_at_0} is $\Q(z)$-equivalent to a companion matrix system, i.e., there exist a matrix $P \in {\rm GL}_m\left(\Q(z)\right)$ and rational functions $q_0,\ldots, q_{m-1} \in \Q(z)$ such that $\puip(P)A P^{-1}=A_{\rm comp}$ where
\begin{equation}\label{eq:ACompanion}
A_{\rm comp}:=\left(\begin{array}{ccccc} 0 & 1 & 0 & \cdots & 0
\\ \vdots & \ddots & \ddots & & \vdots 
\\ \vdots && \ddots& \ddots & 0
\\ 0 &\cdots & \cdots & 0 & 1
\\ q_0 & \cdots & \cdots & \cdots & q_{m-1}
\end{array}\right).
\end{equation}
\end{thm} 

\begin{proof}
We adapt the proof of Birkhoff given in \cite[§1]{Bir30} and of Sauloy given in \cite[Annexe B.2]{Sau00}.  
In order to build such a matrix $P$, we build its rows $\r_1,\ldots,\r_m$. These rows must be linearly independent and must satisfy
\begin{equation}
\label{eq:rows_cyclicLemma}
    \puip\left(\r_i\right)A=\r_{i+1} \quad \text{for} \quad 1\leq i\leq m-1.
\end{equation}
Therefore, we are looking for a vector  $\r\in\Q(z)^m$ such that 
the vectors $\r_1:=\r$, $\r_{i+1}:=\puip\left(\r_i\right)A$, $1\leq i\leq m-1$ form a basis of $\Q(z)^m$. For this purpose, we choose $z_0\in\Q^\star$ not a root of unity
such that 
$A\left(z_0\right), \ldots, A\left(z_0^{p^{m-2}}\right)\in {\rm GL}_m\left(\Q\right)$ (such a $z_0$ exists because the matrix $A$ has finitely many singularities). Since $z_0,z_0^p,\ldots,z_0^{p^{m-1}}$ are distinct, 
we can choose, by a polynomial interpolation process, a vector $\r\in\Q[z]^m$ such that
\begin{equation}\label{eq:interpolation}
\left\lbrace\begin{array}{rcl}
   \r(z_0)&=&\e_1 \\
    \r(z_0^p)&=&\e_2 A(z_0)^{-1} \\
    &\vdots& \\
    \r(z_0^{p^{m-1}})&=&\e_m A(z_0)^{-1}\ldots A\left(z_0^{p^{m-2}}\right)^{-1}
\end{array}\right.
\end{equation}
where $\e_1,\ldots,\e_m$ is the canonical basis of $\Q^m$. 
Write $\r_1:=\r$ and define recursively $\r_{i+1}:=\puip\left(\r_i\right)A$, $1\leq i \leq m-1$. By construction, 
$$\r_i(z_0)=\e_i\,, $$
and the matrix $P$ whose rows are the vectors $\r_1,\ldots, \r_m$ satisfies $P(z_0)={\rm I}_m$. Thus, $P\in {\rm GL}_m\left(\Q(z)\right).$ Write
$$
(q_0,\ldots,q_{m-1}):=\puip(\r_m)AP^{-1}\, .
$$
Then $A_{\rm comp} = \puip(P)A P^{-1}$ is a companion matrix of the form \eqref{eq:ACompanion}.
\end{proof}

\begin{rem} If one chooses $\e_1,\ldots,\e_m$ to be any basis of $\Q^m$, instead of the canonical basis, a different rational interpolation, and another $z_0$ in the proof of Theorem \ref{thm:CyclicVector}, one would obtain a different matrix $P$ and different rational functions $q_0,\ldots,q_{m-1}$. The different companion matrices obtained after this process, with these different choices, are $\Q(z)$-equivalent. Actually, any companion matrix system equivalent to \eqref{eq:Mahler_at_0} can be obtained by this process in the following way. Indeed, let us assume that $P \in {\rm GL}_m(\Q(z))$ is such that $\puip(P)AP^{-1}$ is a companion matrix. Let $z_0$ be chosen as in the proof of Theorem \ref{thm:CyclicVector} with the additional assumption that $P$ is well defined and non-singular at $z_0$. We now let $\e_1,\ldots,\e_m$ denote the rows of $P(z_0)$, instead of the canonical basis. Let $\r$ denote the first row of $P$. Since $\puip(P)AP^{-1}$ is a companion matrix, $\r$ satisfies the interpolation conditions \eqref{eq:interpolation}. Then, let $\r_1,\ldots,\r_m$ be defined from $\r$ and $A$ as in the proof of Theorem \ref{thm:CyclicVector}. One easily checks that $\r_1,\ldots,\r_m$ are the rows of $P$.
\end{rem}

\subsection{Ramification index of vector solutions of Mahler systems}\label{sec:Ramification}

In this section, we study the ramification index of solutions of linear systems of the form \eqref{eq:linear_system}. Let $\mathcal D$ denote the set of integers $d \in \{1,\ldots,p^m-1\}$ such that $p$ and $d$ are relatively prime.

\begin{lem}\label{lem:ramification_step1}
Let $A\in {\rm GL}_m\left(\Q(z)\right)$. There exists an integer $d \in \mathcal D$ such that for any $\lambda \in \Q^\star$ and $\f \in \mathbf K^m$ satisfying 
\begin{equation}
    \label{eq:system_homogene_in_ramification}
\lambda\puip(\f) +\puip(\g) = A \f 
\end{equation}
with $\g\in \Q\left(\left(z^{1/d}\right)\right)^m$, we have $\f \in \Q\left(\left(z^{1/d}\right)\right)^m$. 
\end{lem}

In the sequel, we let $\mathcal D_0 \subset \mathcal D$ denote the set of all such integers $d \in \mathcal D$.
Note that, if $d_1,d_2 \in \mathcal D_0$ then so does $\gcd(d_1,d_2)$. Indeed, if $d_0 := \gcd(d_1,d_2)$, then $\Q((z^{1/{d_0}}))=\Q((z^{1/{d_1}}))\cap\Q((z^{1/{d_2}}))$.
\begin{proof}
By Theorem \ref{thm:CyclicVector} there exists $P \in {\rm GL}_m(\Q(z))$ such that $A_{\rm comp} = \puip(P)A P^{-1}$ is a companion matrix, that is, a matrix of the form \eqref{eq:ACompanion}.
We let $\f,\g, \lambda$ be as in the lemma. We have
$$
\lambda\puip(P\f)  + \puip(P\g) = A_{\rm comp} P\f
$$
so, without loss of generality, we replace $P\f$ with $\f$ and $P\g$ with $\g$, which does not modify the ramification index, and we assume that $A = A_{\rm comp}$. Let $f_1,\ldots,f_m \in \mathbf K$ and $g_1,\ldots,g_m \in \Q\left(\left(z^{1/d}\right)\right)$ be the entries of $\f$ and $\g$ respectively. Since $A$ is a companion matrix, we infer from \eqref{eq:system_homogene_in_ramification} that, for every $i\in\lbrace 1,\ldots, m-1\rbrace$,
\begin{equation}
\label{eq:inhomog_in_ram}
f_{i+1} = \puip(g_i)+\lambda \puip(f_i)= \cdots = \lambda^i\puip^i\left(f_1\right) + \sum\limits_{j=1}^{i} \lambda^{j-1} \puip^{j}(g_{i-j+1}).
\end{equation}
To find an integer $d$ satisfying the conclusion of Lemma \ref{lem:ramification_step1} we first assume that $\g=\boldsymbol{0}$. 
Considering the last row of the system \eqref{eq:system_homogene_in_ramification}, it follows from \eqref{eq:inhomog_in_ram} that  $f_1$ is a solution of the Mahler equation
\begin{equation}\label{eq:eq_homogene_in_ramification}
q_0 y + \lambda q_1 \puip(y) + \cdots + \lambda^{m-1} q_{m-1}\puip^{m-1}(y) - \lambda^m \puip^m(y) = 0 .
\end{equation}
Let $\mathcal H \in \mathbb R^2$ denote the \emph{lower hull} of the set of pairs $\left(p^i,\val(\lambda^{i}q_i)\right)$, $0\leq i\leq m$, with $q_m:=1$. Let $d\in\mathbb{N}^\star$ be the least common multiple of the denominators of the slopes of $\mathcal H$ which are coprime with $p$. If there are no such denominators we write $d:=1$. Since $\val(\lambda^iq_i)=\val(q_i)$, this integer $d$ does not depend on $\lambda$. From \cite[Prop. 2.19]{CDDM18}, $d\in\mathcal{D}$ and any solution $y \in \mathbf K$ of \eqref{eq:eq_homogene_in_ramification} belongs to $\Q\left(\left(z^{1/d}\right)\right)$. Therefore, $f_1\in \Q((z^{1/d}))$.

We now prove that Lemma \ref{lem:ramification_step1} holds with this integer $d$. Let $\g\in \Q((z^{1/d}))^m$. From \eqref{eq:inhomog_in_ram}, we only have to prove that $f_1 \in \Q((z^{1/d}))$. Considering the last row of the system \eqref{eq:system_homogene_in_ramification}, it follows from \eqref{eq:inhomog_in_ram} that  
\begin{equation}\label{eq:eq_inhomogene_in_ramification}
q_0 f_1 + \lambda q_1 \puip(f_1) + \cdots + \lambda^{m-1} q_{m-1}\puip^{m-1}(f_1) - \lambda^m \puip^m(f_1) = g_0 ,
\end{equation}
where $g_0 \in \Q\left(\left(z^{1/d}\right)\right)$ is a $\Q(z)$-linear combination of the $\puip^j(g_i)$, $i,j \in \{1,\ldots, m\}$. First, we prove that $f_1\in {\bf k}$, where ${\bf k}:=\bigcup_{\ell\in \mathbb{N}}\Q\left(\left(z^{1/(dp^\ell)}\right)\right) \subset \K$. The function $f_1\in\mathbf K$ can be written as 
$$
f_1=h_0+h_1
$$
where $h_0\in{\bf k}$ and none of the monomials in the Puiseux expansion of $h_1$ belong to ${\bf k}$. Then, none of the monomials of the Puiseux expansion of $\puip^j(h_1)$, $j\in\lbrace 0,\ldots, m\rbrace$, belong to ${\bf k}$. Hence, since $g_0 \in \Q\left(\left(z^{1/d}\right)\right) \subset {\bf k}$, it follows from \eqref{eq:eq_inhomogene_in_ramification} that $h_1$ is a solution of \eqref{eq:eq_homogene_in_ramification}. From the first part of the proof, $h_1 \in \Q\left(\left(z^{1/d}\right)\right)$. Thus $h_1=0$ and $f_1 = h_0 \in{\bf k}$. Let $\ell_0$ be the smallest integer such that $f_1\in \Q\left(\left(z^{1/(dp^{\ell_0})}\right)\right)$. We assume by contradiction that $\ell_0>0$. From \eqref{eq:eq_inhomogene_in_ramification}, $f_1$ is a $\Q(z)$-linear combination of $g_0$ and the $\puip^j(f_1)$ for $j \in \{1,\ldots, m\}$, which are all elements of $\Q\left(\left(z^{1/(dp^{\ell_0-1})}\right)\right)$. Thus, $f_1\in \Q\left(\left(z^{1/(dp^{\ell_0-1})}\right)\right)$, which provides a contradiction. As a consequence, $\ell_0=0$ and $f_1 \in\Q\left(\left(z^{1/d}\right)\right) $ as wanted.
\end{proof}

\begin{coro}\label{coro:D0}
Assume that the Mahler system \eqref{eq:Mahler_at_0} is regular singular at $0$ and let $\Gtr \in {\rm GL}_m(\mathbf K)$ be such that $\puip(\Gtr)^{-1}A\Gtr$ is a constant matrix. Then $\Gtr$ belongs to 
$$ \bigcap_{d \in \mathcal D_0} {\rm GL}_m\left(\Q\left(\left(z^{1/d}\right)\right)\right)\,.$$
\end{coro}

\begin{proof}
We can assume that $\puip(\Gtr)^{-1}A\Gtr$ is a Jordan matrix. Thus \eqref{eq:Mahler_columns} holds. Let $d\in\mathcal{D}_0$. Using the notations of \eqref{eq:Mahler_columns}, we prove by induction on $i\in\lbrace 1, \ldots, m\rbrace$ that the columns $\gtr_{1},\ldots,\gtr_m$ of $\Gtr$ belong to $\Q\left(\left(z^{1/d}\right)\right)^m$. From \eqref{eq:Mahler_columns} we have
$$
\lambda_1 \puip(\gtr_1)=A\gtr_1\,.
$$
Thus, it follows from Lemma \ref{lem:ramification_step1} applied with $\lambda=\lambda_1$, $\f=\gtr_1$ and $\g=0$ that $\gtr_1\in\Q\left(\left(z^{1/d}\right)\right)^m$. Assume that $i\geq 2$ and that $\gtr_{i-1} \in \Q\left(\left(z^{1/d}\right)\right)^m$. Then, it follows from \eqref{eq:Mahler_columns} and Lemma \ref{lem:ramification_step1} applied with $\lambda=\lambda_i$, $\f = \gtr_i$ and $\g=\epsilon_i\gtr_{i-1}$ that $\gtr_{i} \in \Q\left(\left(z^{1/d}\right)\right)^m$.
\end{proof}

One could be tempted to work with the smallest integer $d \in \mathcal D_0$. However, while Algorithm \ref{Algo:dbis} below returns an integer $d \in \mathcal D_0$, there is no guaranty that this integer is minimal. Thus, in what follows, we shall work with any $d\in \mathcal D_0$.

\subsection{Valuation of vector solutions of Mahler systems}\label{sec:valuation}

In this subsection, we fix an integer $d \in \mathcal D_0$ and we consider
\begin{equation}
\label{eq:def_nu_d}
\nu_d:= \lceil d\val(A)/(p-1) \rceil.
\end{equation}
 We prove that the valuation of vector solutions of systems of the form \eqref{eq:linear_system} is at least $\nu_d/d$ assuming that the valuation of $\g$ is at least $\nu_d/d$.

\begin{lem}\label{lem:valuation}
Let $\lambda \in \Q^\star$ and let $\g \in \Q((z^{1/d}))^m$ be a vector of Puiseux series whose valuation is at least $\nu_d/d$. The valuation at $0$ of a solution $\f \in \Q((z^{1/d}))^m$ of
\begin{equation}\label{eq:MahlerLambdaInhomogene_1}
\lambda \puip(\f) + \puip(\g) = A\f\, ,
\end{equation}
is at least $\nu_d/d$.
\end{lem}

\begin{proof}
From \eqref{eq:MahlerLambdaInhomogene_1}
we have
$$
p\val(\f)\geq \min\left(\val(A)+\val(\f), p\val(\g)\right):=n_0.
$$Two cases occur:
\begin{itemize}
\item If $n_0 = \val(A)+\val(\f)$, then 
$p\val(\f) \geq \val(A) + \val(\f)$ and
$$
\val(\f) \geq \frac{\val(A)}{p-1}\,.
$$
Since $d\val(\f)$ is an integer we have $d\val(\f)\geq \nu_d$, as wanted.
\item If $n_0=p\val(\g)$ then $\val(\f)\geq \val(\g)\geq \nu_d/d$, which concludes.
\end{itemize}
\end{proof}

We then have the following corollary.
\begin{coro}
\label{coro:valuation} Let $d \in \mathcal D_0$.
Suppose that the system \eqref{eq:Mahler_at_0} is regular singular at $0$ and let $\Gtr \in {\rm GL}_m\left(\Q\left(\left(z^{1/d}\right)\right)\right)$ be such that $\puip(\Gtr)^{-1}A\Gtr$ is a constant matrix. Then, $\val(\Gtr)\geq~\nu_d/d$.
\end{coro}
\begin{proof}
Let $d \in \mathcal D_0$. From Corollary \ref{coro:D0}, $\Gtr \in {\rm GL}_m(\Q\left(\left(z^{1/d}\right)\right)$. Then, arguing as in the proof of Corollary \ref{coro:D0}, we prove by induction, using \eqref{eq:Mahler_columns} and Lemma \ref{lem:valuation}, that the valuations of the columns of $\Gtr$ are at least $\nu_d/d$.
\end{proof}


\subsection{Coefficients of vector solutions of Mahler systems}
\label{subsec:solutions}
Let $d \in \mathcal D_0$ be an integer and let $\nu_d$ be defined by \eqref{eq:def_nu_d}. Instead of studying solutions in $\Q((z^{1/d}))^m$ of \eqref{eq:MahlerLambdaInhomogene_1}, we use the operator $\puid : z \mapsto z^d$ to work in the field $\Q((z))$ of Laurent series. To compute the coefficients of such vectors of solutions, we need to \emph{inverse} the Mahler system. We write $\Ainv_d := \phi_d(A)^{-1}$ and we let
$$
\Ainv_d := \sum_{n \geq d\val\left(A^{-1}\right)} \Ainv_{d,n} z^n
$$
denote the Laurent expansion of $\Ainv_d$. 
Let $\f \in \Q((z))^m$ be a solution of the linear system
$$
\lambda \f = \Ainv_d\puip(\f)\, ,
$$
for some $\lambda \in \Q^\star$. By Lemma \ref{lem:valuation}, the valuation of $\f$ is at least $\nu_d$. We write $\f=\sum_{n\geq \nu_d} \f_n z^n$, $\f_n \in \Q^m$, and $\f_n = 0$ if $n < \nu_d$. Then, for every $n \in \Z$, we have
$$
\lambda \f_n = \sum_{(k,\ell)\,:\, k +p\ell=n}\Ainv_{d,k}\f_\ell\, .
$$
Write
\begin{equation}\label{eq:mud_nud}
\mu_d := \lceil -d\val\left(A^{-1}\right)/(p-1)\rceil\,.
\end{equation}
Since $AA^{-1} = {\rm I}_m$, we have $\val(A)+\val(A^{-1}) \leq 0$ so $\nu_d\leq \mu_d$.
The vectors $\f_\ell$ which are taken into account on the right-hand side of the equation have an index $\ell\leq \frac{n-d\val(A^{-1})}{p}$. Then, if $n > \mu_d$, we have $\frac{n-d\val(A^{-1})}{p}<n$. Thus, $\f_n$ is uniquely determined by the vectors $\f_\ell$, $\ell < n$. Moreover, the coefficients of the vectors $\f_\ell$, $\nu_d \leq \ell \leq \mu_d$, are solutions of some linear equations depending on $\lambda$ and $\Ainv_d$. Thus, the problem of determining $\f$ can be transformed into a finite dimensional problem. To capture this we introduce the following map:
\begin{eqnarray*}
\pi_d:\quad  \Q((z))^m & \to & \Q^{m(\mu_d-\nu_d+1)}
\\
     \sum_{n\in\Z} \g_nz^n & \mapsto&  \begin{pmatrix}  \g_{\nu_d} \\ \vdots \\ \g_{\mu_d}\end{pmatrix}\,.
\end{eqnarray*}
Then we define two block matrices
\begin{align*}
M_d&:=(B_{d,i-pj})_{\nu_d \leq i,j\leq \mu_d},\ \text{ and }\ \\ \nonumber N_d&:=(B_{d,i-pj})_{\begin{subarray}{l}
d\val\left(A^{-1}\right)+p\nu_d \leq i \leq \nu_d-1,\ \nu_d \leq j \leq \mu_d\,. \end{subarray}}
\end{align*}
We proceed to check that the map $\pi_d$ and these matrices are well defined.
Since $\nu_d \leq \mu_d$, $\pi_d$ and the matrix $M_d$ are well defined. Now, $\nu_d < \mu_d$ if and only if $\nu_d<~-d\val\left(A^{-1}\right)/(p-1)$. In that case, $d\val(A^{-1}) + p\nu_d \leq \nu_d-1$ and the matrix $N_d$ is well defined. If $\nu_d=\mu_d$, then $d\val(A^{-1}) + p\nu_d > \nu_d-1$ and the matrix $N_d$ shall be considered as a matrix with no rows.

\begin{lem}\label{lem:solutioninhomogène}
Let $\gamma \in \Q^\star$ and let $\h \in \Q((z))^m$ be a vector of Laurent series whose valuation is at least $\nu_d$. If $\f \in \Q((z))^m$ is a solution of
\begin{equation}\label{eq:MahlerLambdaInhomogene}
\gamma \f + \h = \Ainv_d \puip(\f)\, ,
\end{equation}
then,
$$
N_d\pi_d(\f)= 0,\ \text{ and } \ M_d\pi_d(\f) = \gamma \pi_d(\f)+\pi_d(\h)\, .
$$
\end{lem}
Note that, if $\mu_d=\nu_d$ and $N_d$ is a matrix with no rows, the condition $N_d\pi_d(\f)=0$ is automatically satisfied for every $\f$. We shall consider that, in that case, the right-kernel of $N_d$ is the whole space $\Q^{m\left(\mu_d-\nu_d+1\right)}$.

\begin{proof}
From \eqref{eq:MahlerLambdaInhomogene}, we have $\puip\left(\f\right)=\puid(A)\left(\gamma\f+\h\right)$. It follows that $p\val\left(\f\right)\geq d\val\left(A\right)+\min\left(\val\left(\f\right),\val\left(\h\right)\right)$. Arguing as in the proof of Lemma \ref{lem:valuation}, one checks that the valuation of $\f$ is at least $\nu_d$. Write $$
\f = \sum_{n\geq \nu_d} \f_n z^n,\quad \h = \sum_{n\geq \nu_d} \h_n z^n, \qquad \f_n, \h_n\in \Q^m,
$$
and $\f_n,\h_n:=0$, when $n < \nu_d$. The series $\f$ is a solution of \eqref{eq:MahlerLambdaInhomogene} if and only if
\begin{equation}\label{eq:recurrenceFinhomogene}
\forall n \in \Z,\quad \gamma \f_n +\h_n= \sum_{(k,\ell)\, : \, k+p\ell=n} \Ainv_{d,k} \f_\ell\, .
\end{equation}
When $n < \nu_d$, the left-hand side is $0$. If $\ell > \mu_d$, then $n-p\ell < \val(\Ainv_d)=d\val(A^{-1})$. Thus, we have
$$
\forall n < \nu_d,\quad \sum_{\ell = \nu_d}^{\mu_d} \Ainv_{d,n-p\ell} \f_\ell = 0\, . 
$$
In particular, $N_d\pi_d\left(\f\right)=0$.
Now, looking at equation \eqref{eq:recurrenceFinhomogene} for all $n$, $\nu_d\leq n\leq \mu_d$, we have, similarly,
$$
M_d\pi_d\left(\f\right) = \pi_d(\gamma\f+\h)= \gamma \pi_d\left(\f\right)+\pi_d\left(\h\right) \,,
$$
as wanted.
\end{proof}

\section{A characterisation of regular singular Mahler systems at \texorpdfstring{$0$}{0}}
\label{sec:preuvetheorem}

In Section \ref{sec:Ramification_valuation}, we studied the vector solutions of Mahler systems of the form \eqref{eq:linear_system} in $\mathbf K^m$. We computed the possible ramification indexes, a lower bound for the valuation of such solutions, and proved that their coefficients must satisfy certain linear relations over $\Q$. Let $d \in \mathcal D_0$ be an integer. The conclusion of Lemma \ref{lem:solutioninhomogène} invites us to define the following vector spaces:
$$
\intersd^{+} := \bigcap_{n \in \N} M_d^n \ker(N_d)\,,\quad \intersd^{-}:= \bigcap_{n \in \N} \ker (N_dM_d^n)
$$
and 
$$
\intersd = \intersd^{+} \bigcap \intersd^{-}\, .
$$
The main result of this paper states as follows.

\begin{thm}\label{thm:Espace}  
The three following propositions are equivalent:
\begin{enumerate}
    \item\label{1} The Mahler system \eqref{eq:Mahler_at_0} is regular singular at $0$,
    \item\label{2} $\dim \intersd \geq m$ for some integer $d \in \mathcal D_0$, 
    \item\label{item3} $\dim \intersd = m$ for every integer $d \in \mathcal D_0$.
\end{enumerate}
 In that case, the system is $\Q((z^{1/d}))$-equivalent to a constant system for every $d \in \mathcal D_0$.
\end{thm}

Thus, to prove that some Mahler system is regular singular at $0$, one only needs to check point (\ref{2}). Before proving Theorem \ref{thm:Espace} we need two lemmas. The first one is about the linear independence of vector solutions of linear Mahler systems. The second one gives an implicit characterisation of the vector space $\intersd$.

 \begin{lem}
\label{lem:indeplineairecolonnes}
Let $T$ be a matrix with entries in $\Q((z^{}))$ and $D$ be a constant invertible matrix such that
\begin{equation}
\label{eq:TD=B puip(T)}
T D = \Ainv_d \puip(T)\ .    
\end{equation}
If the columns of $T$ are linearly dependent over $\Q((z^{}))$, then they are linearly dependent over $\Q$.
\end{lem}

\begin{proof}
Let $P$ be a constant invertible matrix such that the matrix $PDP^{-1}$ is upper triangular. Then, $(TP^{-1})(PDP^{-1}) = \Ainv_d \puip(TP^{-1})$. Thus, without loss of generality, we replace $T$ with $TP^{-1}$ and we assume that $D$ is upper triangular. We can also assume that the first column of $T$ is nonzero, otherwise the conclusion of the lemma is immediate. Let $a$ be the least integer such that the first $a$ columns of the matrix $T$ are linearly dependent over $\Q((z^{}))$. By assumption, $a \geq 2$. There exists a column vector $\g:=(g_1,\ldots,g_{a-1},1,0,\ldots,0)^\top \in \Q((z))^{m}$, $m\in\N$, such that 
\begin{equation}\label{eq:GammaG=0}
T\g=0\, .
\end{equation}
Mutliplying \eqref{eq:TD=B puip(T)} by $\puip(\g)$ one obtains
\begin{equation}\label{eq:GammaCpuip(G)=0}
T D \puip(\g) = \Ainv_d \puip(T)\puip(\g)=\Ainv_d \puip\left(T\g\right) = 0.
\end{equation}
Since $D$ is upper triangular, the $a$th coordinate of $D\puip(\g)$ is some eigenvalue $\eta \in \Q^\star$ of $D$ and the $m-a$ last coordinates of $D\puip(\g)$ are zero. By minimality of $a$, we infer  from \eqref{eq:GammaG=0} and \eqref{eq:GammaCpuip(G)=0} that
$$
D \puip(\g) = \eta \g\, .
$$
From \cite[Thm. 3.1]{Ni97}, $\g \in \Q^{m}$ and Equation \eqref{eq:GammaG=0} provides a linear relation over $\Q$ between the columns of $T$, as wanted. 
\end{proof}

\begin{lem}\label{lem:X_max}
Let $d \in \mathcal D_0$. The vector space $\intersd$ is the largest subspace of $\ker(N_d)$ on which $M_d$ acts as an isomorphism.
\end{lem}
\begin{proof}
By definition, $\intersd \subset \ker(N_d)$ and $\intersd$ is invariant under the action of $M_d$. Since $\intersd$ is finite dimensional, to prove that $M_d$ acts as an isomorphism on $\intersd$ we only have to prove that $\ker(M_d) \cap \intersd = \{0\}$. Let $\x \in \ker(M_d) \cap \intersd $, and let $s\times s$ denote the size of $M_d$. Then $\ker(M_d^s)=\ker(M_d^{s+1})$. Since $\x \in \intersd \subset M_d^s\ker(N_d)$, there exists  $\y \in \ker(N_d)$ such that $\x = M_d^s\y$. Then, $M_d^{s+1}\y=M_d\x=0$. Thus, $\y \in \ker(M_d^{s+1})=\ker(M_d^s)$ and $\x=M_d^s\y=0$. It follows that $\ker(M_d) \cap \intersd =\{ 0\}$.

Now, let $\mathfrak V \subset \ker(N_d)$ be a vector space on which $M_d$ acts as an isomorphism. On the one hand, $M_d^n\mathfrak V = \mathfrak V \subset \ker(N_d)$ for every $n \in \N$. Thus, $\mathfrak V \subset \ker(N_dM_d^n)$ for every $n \in \N$. On the other hand, $\mathfrak V = M_d^n\mathfrak V \subset M_d^n\ker(N_d) $ for every $n \in \N$. Therefore, $\mathfrak V \subset \intersd$.
\end{proof}

We are now able to prove Theorem \ref{thm:Espace}.

\begin{proof}[Proof of Theorem \ref{thm:Espace}]
Consider the following proposition :
\medskip
\begin{enumerate}
 \setcounter{enumi}{3}
    \item\label{4} $\dim \intersd\geq m$ for every $d \in \mathcal D_0$.
\end{enumerate}
\medskip
We prove that (\ref{1}) implies (\ref{4}), that (\ref{4}) implies (\ref{item3}) and that (\ref{2}) implies (\ref{1}). Since (\ref{item3}) trivially implies (\ref{2}), this shall prove Theorem \ref{thm:Espace}.

Let $d\in\mathcal{D}_0$ and suppose that the system is regular singular at $0$. Then, it follows from Corollaries \ref{coro:D0} and \ref{coro:valuation} that there exists  $\Gtr \in {\rm GL}_m\left(\Q((z^{1/d}))\right)$ such that $\Lambda:=\puip(\Gtr)^{-1}A \Gtr$ is a constant matrix and $\val(\Gtr)\geq \nu_d/d$.
Write $\T:=\puid(\Gtr)$ and recall that $\Ainv_d:=\puid(A^{-1})$. We have
\begin{equation} \label{eq:MahlerGamma_puid}
\T \Const^{-1} = \Ainv_d\puip\left(\T\right)\, .
\end{equation}
We can assume that $\Const^{-1}$ is a Jordan matrix that is
$$
 \Const^{-1}:= \begin{pmatrix}
J_{s_1}(\gamma_1) &  &  & \\
 & J_{s_2}(\gamma_2) & &  \\
 & & \ddots &   \\
 & &  & J_{s_r}(\gamma_r)
\end{pmatrix}   
$$
where $\gamma_1,\ldots,\gamma_r$ are nonzero algebraic numbers and $J_{s_i}(\gamma_i)$ is the Jordan block of size $s_i$ associated with the eigenvalue $\gamma_i$.
Let
$$
\t_{1,1},\ldots,\t_{1,s_1},\t_{2,1},\ldots,\t_{2,s_2},\ldots,\t_{r,1},\ldots,\t_{r,s_r}
$$
denote the columns of $\T$ indexed according to the Jordan block decomposition of $\Const^{-1}$. We infer from \eqref{eq:MahlerGamma_puid} that the columns of $\T$ satisfy
\begin{eqnarray}
    \label{eq:colomnesGamma}
   \gamma_i \t_{i,1}&=&\Ainv_d\puip(\t_{i,1})\quad 1\leq i \leq r
    \\ \nonumber \gamma_i\t_{i,j}+\t_{i,j-1}&=&\Ainv_d\puip(\t_{i,j}) \quad 1\leq i \leq r,\, 2\leq j\leq s_i.
\end{eqnarray}
It follows from \eqref{eq:colomnesGamma} and Lemma \ref{lem:solutioninhomogène} applied with $\g = \t_{i,j-1}$ that
$\pi_d(\t_{i,j}) \in \ker(N_d)$ and that
\begin{equation}\label{eq:Md_colonnes}
M_d\pi_d(\t_{i,j}) = \gamma_i\pi_d(\t_{i,j}) + \pi_d(\t_{i,j-1}),
\end{equation}
for every $i,j$, $1\leq i \leq r$, $1\leq j\leq s_i$, 
where $\t_{i,0}=\boldsymbol 0$ for every $i$. Let $\V$ denote the vector space spanned by the vectors $\pi_d(\t_{i,j})$,  $1\leq i \leq r$, $1\leq j\leq s_i$. Then, $\V \subset \ker(N_d)$. It immediately follows from \eqref{eq:Md_colonnes} that $\V$ is invariant under the left multiplication by $M_d$.  We prove that the $m$ vectors $\pi_d(\t_{i,j})$ are linearly independent over $\Q$. By contradiction, assume that they are not linearly independent. Let $k$ be the least integer such that the image by $\pi_d$ of the first $k$ columns of $\T$ are linearly dependent. There exists a non-zero vector $\boldsymbol{\lambda}:= \left(\lambda_1, \ldots,\lambda_{k-1}, 1, 0, \ldots, 0\right)^\top\in\Q^m$ such that $n_0:=\val(\T\boldsymbol{\lambda})>\mu_d$. Multiplying \eqref{eq:MahlerGamma_puid}  with $\boldsymbol{\lambda}$ and looking at the valuations on both sides gives
$$
\val\left(\T\Const^{-1}\boldsymbol{\lambda}\right) \geq d\val\left(A^{-1}\right) +p n_0 \geq \frac{-d\val\left(A^{-1}\right)}{p-1}+p(n_0-\mu_d)\,.
$$
Therefore, we have
\begin{equation}
\label{eq:valuationcontrad}
    \val\left(\T\Const^{-1}\boldsymbol{\lambda}\right)\geq \mu_d + p(n_0-\mu_d)\, .
\end{equation}
Since $\Const^{-1}$ is upper triangular (because we assumed that it is a Jordan matrix), the vector $\Const^{-1}\boldsymbol{\lambda}$ is also of the form $\left(\eta_1, \ldots, \eta_k, 0, \ldots, 0\right)^\top$ with $\eta_k\neq 0$. By \eqref{eq:valuationcontrad}, we have $ \val\left(\T\Const^{-1}\boldsymbol{\lambda}\right)>\mu_d$. Thus,  $\pi_d\left(\T\Const^{-1}\boldsymbol{\lambda}\right)=0$. Then, by minimality of $k$, $\Const^{-1}\boldsymbol{\lambda}= \eta_k \boldsymbol{\lambda}$.
Thus, $\val\left(\T\Const^{-1}\boldsymbol{\lambda}\right) = \val\left(\eta_k\T\boldsymbol{\lambda}\right) = n_0$. Then, from the inequality \eqref{eq:valuationcontrad}, we have $\mu_d\geq n_0$, which is a contradiction. Thus, the $m$ vectors $\pi_d(\t_{i,j})$ are linearly independent and they form a basis of $\V$. Now, from \eqref{eq:Md_colonnes}, the representation of the action of $M_d$ on $\V$ in the basis  $(\pi_d(\t_{i,j}))_{1\leq i \leq r,\,1\leq j \leq s_i}$ is just the matrix $\Const^{-1}$. Since it is nonsingular, $M_d$ acts as an isomorphism on $\V$.
Hence, by Lemma \ref{lem:X_max}, $\V \subset \intersd$ and
$$
\dim \intersd \geq \dim \V = m\, .
$$
Thus (\ref{1}) implies (\ref{4}).

We let $d\in\mathcal{D}_0$ and assume that $\dim \intersd := n \geq m$. We prove that $n= m$ and that the Mahler system \eqref{eq:Mahler_at_0} is regular singular at $0$. We deduce that (\ref{4}) implies (\ref{item3}) and (\ref{2}) implies (\ref{1}). Let $\boldsymbol e_1,\ldots,\boldsymbol e_{n}$ denote a basis of $\intersd$ and let $E$ be the $m(\mu_d-\nu_d+1)\times n$ matrix whose columns are $\boldsymbol e_1,\ldots,\boldsymbol e_{n}$. Since $M_d$ acts as an isomorphism on $\intersd$, there exists $R\in{\rm GL}_{n}\left(\Q\right)$ such that 
\begin{equation}\label{eq:ME=ER}
M_dE=ER.
\end{equation}
We make a block decomposition of $E$ into $\mu_d-\nu_d+1$ matrices $E_{\nu_d},\ldots,E_{\mu_d}$ of size $m\times n$: 
$$
E = \left(\begin{array}{c} E_{\nu_d} \\ \hline \vdots \\\hline E_{\mu_d} \end{array}\right)\, .
$$
We then define by induction on $j>\mu_d$ a matrix $E_j$, setting
\begin{equation}
    \label{eq:definitionE}
E_j = \left( \sum_{(k,\ell)\, : \, k+p\ell=j} \Ainv_{d,k} E_\ell\right)R^{-1}
\end{equation}
where we recall that $\sum_{n\in \Z} \Ainv_{d,n}z^n=\Ainv_d=\puid(A)^{-1}$. Since $j > \mu_d$, the matrices $E_\ell$ contributing to the right-hand side of the equality have an index $\ell < j$. Hence the matrices $E_j$ are well defined. If $j < \nu_d$, we define $E_j := 0$. We stress that \eqref{eq:definitionE} actually holds for any $j \in \Z$ :
\begin{itemize}
    \item by definition, it holds when $j > \mu_d$ ; \item when $\nu_d \leq j \leq \mu_d$, it follows from the fact that $ER= M_dE$ ;
    \item when $j < \nu_d$, it follows from the fact that $N_dE=0$, for $\intersd\subset \ker(N_d)$.
\end{itemize}
We now write $\Gam := \sum_{j\geq \nu_d} E_jz^j$. It is a matrix with $m$ rows, $n$ columns and entries in $\Q((z))$. It follows from \eqref{eq:definitionE} that 
\begin{equation}\label{eq:GammaR}
\Gam R=\Ainv_d \puip(\Gam)\, .
\end{equation}
Since $\boldsymbol e_1,\ldots,\boldsymbol e_{n}$ is a basis of $\intersd$, the columns of $\Gam$ are linearly independent over $\Q$. It follows from Lemma \ref{lem:indeplineairecolonnes} that they are linearly independent over $\Q((z))$ so $n \leq m$. Thus $n=m$ and the matrix $\Gam$ is invertible. In particular (\ref{4}) implies (\ref{item3}). Now, let us define $\Gtr:=\phi_{1/d}(\Gam)$. It follows from \eqref{eq:GammaR} that 
\begin{equation}\label{eq:TR}
\puip(\Gtr)^{-1}A\Gtr = R^{-1} \in {\rm GL}_m\left(\Q\right).
\end{equation}
Thus, the system is regular singular at $0$ and $\Q((z^{1/d}))$-equivalent to a constant system with matrix $R^{-1}$. The matrix $\Gtr$ is an associated gauge transformation. This proves that (\ref{2}) implies (\ref{1}).
\end{proof}

Let ${\bf k} \subset \Q$ denote a number field such that $A \in {\rm GL}_m({\bf k}(z))$. Then the vector space $\intersd$ is defined over ${\bf k}$ and the matrices $E$ and $R$ in the proof of Theorem \ref{thm:Espace} can be chosen with their entries in ${\bf k}$.

\begin{coro}\label{coro:cdn}
Let ${\bf k}\subset \Q $ be a number field and $A \in {\rm GL}_m({\bf k}(z))$. The system \eqref{eq:Mahler_at_0} is regular singular at $0$ if and only if it is $\widehat{{\bf k}(z)}$-equivalent to a matrix in ${\rm GL}_m({\bf k})$, where
$$
\widehat{{\bf k}(z)}:=\bigcup_{d \in \N} {\bf k}\left(\left(z^{1/d}\right)\right)\,,
$$
is the field of Puiseux series with coefficients in ${\bf k}$.
\end{coro}

\section{A concrete algorithm for Theorem \ref{thm:algo}}\label{sec:Algo}

Theorem \ref{thm:Espace} gives the description of a vector space whose dimension characterises the regular singularity at $0$ of a Mahler system \eqref{eq:Mahler_at_0}. In this section we show that the construction of Theorem \ref{thm:Espace} is algorithmic. This provides a proof of Theorem \ref{thm:algo}. Then, we discuss the complexity of this algorithm.

\begin{rem}
In what follows, when discussing the complexity of our algorithms, we shall count the number of operations in $\Q$. However, if ${\bf k} \subset \Q$ is the smallest number field such that $A \in {\rm GL}_m({\bf k}(z))$, our operations are done with elements of ${\bf k}$. To count the number of operations over the rational numbers, one should add a factor $\mathcal O(\M([{\bf k}: \mathbb Q]))$, where $[{\bf k}: \mathbb Q]$ is the degree of ${\bf k}$ over $\mathbb Q$, to the bounds we give.
\end{rem}

\subsection{Description of an algorithm computing a ramification index}

To apply the result of Theorem \ref{thm:Espace}, we first have to find an element $d$ in the set $\mathcal{D}_0$. This integer is related to the valuations at $0$ of the entries of a companion matrix $A_{{\rm comp}}$, $\Q(z)$-equivalent to $A$, which we obtain thanks to the cyclic vector lemma (Theorem \ref{thm:CyclicVector}).

Recall that, from the Cauchy's Theorem (see \cite[Th. 27,2]{Mar66}), the modulus of any root of a nonzero polynomial 
$$
f:=f_0+f_1z+f_2z^2+\cdots + f_h z^h \quad \text{with}\quad f_0,\ldots,f_{h-1}\in\mathbb{C}, f_h\in\mathbb{C}\setminus\lbrace 0\rbrace
$$
is smaller than $1$ plus the max of $\frac{\vert f_k\vert}{\vert f_h \vert}$, $0\leq k\leq h-1$. However, number fields are not necessarily invariant under the map $x \mapsto \vert x \vert$. To stay in the initial base field, we shall not consider directly the absolute value. Let ${\bf k}$ denote a number field such that $A \in {\rm GL}_m({\bf k}(z))$. We fix an embedding ${\bf k} \hookrightarrow \C$. 
We can obtain an upper bound $ V(\xi)\in\mathbb Q$ for the absolute value of any $\xi \in {\bf k}$.
Then, for $f=f_0+f_1z + \cdots + f_hz^h$, $f_i \in {\bf k}$, $f_h\neq 0$, we write
$$
\Vert f \Vert := 1 + \max\left\{V\left(\frac{ f_k}{ f_h }\right),\, 0\leq k\leq h-1 \right\}\geq   1 + \max\left\{ \left\vert \frac{f_k}{f_h}  \right\vert   ,\, 0\leq k\leq h-1 \right\} > 1,
$$
if $h \geq 1$ and $\Vert f \Vert =2$ otherwise. Hence, $\Vert f \Vert\in\mathbb Q$ is greater than the modulus of every root of $f$. We assume that $V(\xi)$ is computable in $\mathcal O(1)$ for any $\xi \in \mathbf k$ so that $\Vert f \Vert$ is computable in $\mathcal O(\deg(f))$.

The following algorithm takes a Mahler system as input, computes a companion matrix $A_{\rm comp}$ such that the systems associated with $A$ and $A_{\rm comp}$ are $\Q(z)$-equivalent, and returns the valuations of the last row of this companion matrix.

\medskip

\begin{algorithm}[H]\label{Algo:Cyclicbis}
\SetAlgoLined
\KwIn{$A\in{\rm GL}_m\left(\Q(z)\right)$, $p\in\mathbb{N}_{\geq 2}$.}
\KwOut{The valuations of the last row of a companion matrix $\Q(z)$-equivalent to $A$.}
Compute $f$ the lcm of the denominators of the entries of $A$.
\\  Write $\widetilde{A}=fA \in{\rm GL}_m\left(\Q[z]\right)$.
\\  Consider $z_0:=\max\left(\Vert f \Vert,\Vert \det(\widetilde{A})\Vert \right)$. 
\\
Compute a solution $\r \in \Q[z]^m$ of \eqref{eq:interpolation} by Newton interpolation.
\\
Let $P$ be the matrix whose rows are $\r_1:=\r$, $\r_{i+1}:=\puip(\r_i)A$, $1\leq i \leq m-1$.
 \\
 \Return the valuation of the entries of $\puip(\r_m)AP^{-1}$.
 \caption{Find the valuation of the entries of the last row of $A_{\rm comp}$}
\end{algorithm}

\medskip
It is clear, from the proof of Theorem \ref{thm:CyclicVector}, that the matrix $\puip(P)AP^{-1}$ is a companion matrix and that its last row is $\puip(\r_m)AP^{-1}$. Now, the following algorithm finds an element of $\mathcal{D}_0$ -- though not necessarily the smallest-- as it was done in the proof of Lemma \ref{lem:ramification_step1}.

\medskip

\begin{algorithm}[H]\label{Algo:dbis}
\SetAlgoLined
\KwIn{$A\in {\rm GL}_m\left(\Q(z)\right)$, $p\in\mathbb{N}_{\geq 2}$.}
\KwOut{An integer $d \in \mathcal D_0$}
Compute $(v_0,\ldots,v_{m-1})$ the valuations of the last row of a companion matrix $\Q(z)$-equivalent to $A$, with Algorithm \ref{Algo:Cyclicbis}.
\\
Compute the lower hull $\mathcal H$ of the set of pairs $\left(p^i,v_i\right)$, $0\leq i\leq m$, with $v_m:=0$.
\\ 
Compute the set $\mathcal S$ of denominators of the slopes of $\mathcal H$ which are coprime with $p$.
\\
\Return ${\rm lcm}(\mathcal S)$.
 \caption{Find some integer $d \in \mathcal D_0$}
\end{algorithm}

\medskip

We compute an upper bound for the complexity of Algorithm \ref{Algo:dbis}. Let us first recall some notations and results.
 Given a $n>0$, we let $\M(n)$ denote the complexity of the product of two polynomials of degree at most $n$, and $\MM(n)$ denote the complexity of the product of two matrices with at most $n$ rows and $n$ columns.
 Let $C\in \mathcal M_m\left(\Q[z]\right)$ with $\det(C)\neq 0$ and let $\delta:=\deg(C)$. The complexity of computing
\begin{itemize}
    \item the determinant of $C$ is $\mathcal O\left({\MM}(m) {\M}(\delta)\left(\log(m)\right)^2\right)$, see \cite{Sto03};
    \item the product $\v C^{-1}$ is $\mathcal O\left(\MM(m)\M(\delta)\log(m)\log(\delta)\right)$, assuming that we know some point at which $C$ is invertible and that the degree of $\v \in \Q[z]^m$ is at most $\delta$, see \cite[Cor. 16]{Sto03};
    \item the inverse of $C$ is $\mathcal O\left(m^2 {\rm M}(m\delta)\log(m\delta)\right)$, see \cite{ZLS15}.
\end{itemize}

\begin{prop}\label{prop:ComplexityCyclic}
The complextiy of Algorithm \ref{Algo:dbis} is
\begin{equation*}
    \label{eq:complexity_cyclic}
\mathcal O\left(\MM(m)\log(m){\M}\left(u\right)\log(u)\right)\quad \mbox{with}\quad u:= (m +\deg(A))p^{m}.
\end{equation*}
\end{prop}

\begin{proof}
We start by computing an upper bound for the complexity of Algorithm \ref{Algo:Cyclicbis}. Assume first that the matrix $A$ has its entries in $\Q[z]$. 
The complexity of computing $\det(A)$ is 
$$
\mathcal O\left(\MM(m)\M(\delta)\left(\log(m)\right)^2\right)
$$
where $\delta=\deg(A)$. 
Then, since $f=1$ here and since $\det(A)$ is a polynomial of degree $\mathcal O(m\delta)$, the computation of $z_0$, which is equal here to $\max\{2,\Vert \det(A) \Vert\}$, can be done with
$\mathcal O(m\delta)$ operations.
To obtain a solution $\r$ of \eqref{eq:interpolation}, we first need to compute the matrices
\begin{equation}
\label{eq:compute_product}
\left(A\left(z_0^{p^k}\right)\ldots A\left(z_0^{p^2}\right) A\left(z_0^p\right) A\left(z_0\right)\right)^{-1},\ 1 \leq k \leq m-2\, .   
\end{equation}
The complexity of taking the $p$th power of a number is $\mathcal O(\log(p))$, thus computing $z_0,z_0^p,\ldots, z_0^{p^{m-2}}$ necessitates $\mathcal O(m\log(p))$ operations. A straightforward evaluation of a polynomial with degree $\ell$ at $n$ points necessitates 
$
\mathcal O\left(n\ell\right)\,
$
operations.
Since the $m^2$ entries of $A(z)$ are polynomials with degree $\delta$, the complexity of computing the matrices $A(z_0),A(z_0^p),\ldots,A(z_0^{p^{m-2}})$ is
$\mathcal O\left(m^3\delta\right)$. We now have to compute $m$ products and inverses of these constant matrices. To sum up, the computation of \eqref{eq:compute_product} can be done with
$$\mathcal O\left(m\log(p)+ m^3\delta+m\MM(m)\right)$$
operations. Then, we compute each of the $m$ entries of $\r$ by doing a Newton interpolation at $m$ points. There, the complexity is
$$\mathcal O(m\M(m)\log(m))$$ (see \cite{BS05}). We use that $\r=\r_1$ and $\r_{k+1}=\puip(\r_k)A$ for every $k$, $1\leq k \leq m-1$, to compute the rows $\r_1,\ldots,\r_m$ of $P$. In particular, 
$$
\deg(\r_k)\leq (m+\delta)p^{k-1}\,.
$$
The computation of $\puip(\r_k)A$ necessitates $m^2$ sums and products of polynomials with degree at most $(m+\delta)p^{k}$. Thus, $\mathcal O\left(m^2 \M\left((m+\delta)p^{k}\right)\right)$ operations suffice to compute $\r_{k+1}$ from $\r_k$.
Hence, once $\r$ is known, one may compute the matrix $P$ with
$$
\mathcal O\left(m^2 \sum_{k=1}^{m-1}\M\left((m+\delta)p^{k}\right)\right)
$$
operations. Then, the complexity of computing $\puip\left(\r_m\right)A$ is
$$
\mathcal{O}\left(m^2\M\left(u\right)\right)\,,
$$
where $u:=(m+\delta)p^m$, 
and the the one of computing $\puip\left(\r_m\right)A P^{-1}$ is
\begin{equation}
    \label{eq:complexity_cyclic_lower}
 \mathcal O\left(\MM(m)\log(m)\M\left(u\right)\log(u)\right) \,. 
\end{equation}
Since \eqref{eq:complexity_cyclic_lower} is greater than the complexity of all the previous steps in Algorithm \ref{Algo:Cyclicbis}, the complexity of Algorithm \ref{Algo:Cyclicbis} is  \eqref{eq:complexity_cyclic_lower}, when $A$ is a matrix with entries in $\Q[z]$.
Assume now that $A$ has rational coefficients. Write $\widetilde A=f A$, with $f \in \Q[z]$ the least common multiple of the denominators of the entries of $A$. Then, by definition, $\deg(\widetilde{A}) \leq \deg(A)$.
Now, the operations with $A=1/f\widetilde{A}$ have the same complexity as the ones with $\widetilde{A}$ and the cost of the computation of $f$ and $\widetilde{A}$ is negligible compared to \eqref{eq:complexity_cyclic_lower}.  Thus, the complexity of Algorithm \ref{Algo:Cyclicbis} is \eqref{eq:complexity_cyclic_lower} for any matrix $A$. Then, the complexity of computing the lower hull in Algorithm \ref{Algo:dbis} is negligible compared to \eqref{eq:complexity_cyclic_lower}. This ends the proof.
\end{proof}

\subsection{Description of the algorithm of Theorem \ref{thm:Espace}}

The following algorithm tests if a given Mahler system is regular singular at $0$.

\medskip
\begin{algorithm}[H]\label{Algo:dfixed}
\caption{Test for the regular singularity of a Mahler system at $0$}
\KwIn{$A\in {\rm GL}_m\left(\Q(z)\right)$, $p\in\mathbb{N}_{\geq 2}$ and the order $n \geq 0$ of truncation.}
\KwOut{Whether or not the system \eqref{eq:Mahler_at_0} is regular singular at $0$ and in that case the constant matrix $\Const$ to which it is equivalent and a truncation of an associated gauge transformation $\Gtr$ at order $n$.}
Compute $d$ with Algorithm \ref{Algo:dbis}.
\\ Compute $\nu_d,\mu_d,M_d,N_d$.
\\ Set $t:=\lceil \log_2(m(\mu_d-\nu_d+1))\rceil$.
\\ \For{$j$ from $1$ to $t$}{ Compute $M_d^{2^j}$.}
 Set $\I_0:= \ker(N_d)$
\\ \For{$\ell$ from $1$ to $t$}{
 Set $\I$ to $\lbrace \x\in\I\mid M_d^{2^{t-\ell}}\x\in \I\rbrace$.}
Set $\inters := M_d^{2^t}\I$.
\\ \If{$\dim \inters = m$}{
   From a basis of $\inters$, compute $R$ and $E_{\nu_d},\ldots,E_{\mu_d}$ as in the proof of Theorem \ref{thm:Espace}.
   \\  \For{$j$ \emph{from} $\mu_d+1$ \emph{to} $\max\{\mu_d+1;dn\}$}{
  Compute $E_j$ from \eqref{eq:definitionE}.}
  Define $\Const:=R^{-1}$.
  \\ \Return``True'', $\Const$ and $\sum_{j =\nu_d}^{dn} E_{j}z^{j/d}$. 
 }
 \Else{ \Return ``False''.}
\end{algorithm}

\medskip
Then, Theorem \ref{thm:algo} is a consequence of the following proposition that we will prove in Section \ref{sec:proof-complexity}.

\begin{prop}
\label{prop:complexity} Algorithm \ref{Algo:dfixed} satisfies the hypothesis of Theorem \ref{thm:algo}.
Apart from the computation of the Puiseux expansion of $\Gtr$, the complexity of Algorithm \ref{Algo:dfixed} is
$$
\widetilde{\mathcal O}\left(m\MM(m)\M((m+\delta)p^m)+mp^m\MM(mv) \right)
$$
where $\delta:=\deg(A)$ and $v:=-(\val(A)+\val(A^{-1}))+1\geq 1$.
\end{prop}

In \cite{Ro20} the author explained how to find the eigenvalues of a constant matrix $\mathcal H$-equivalent to a Mahler system, and the dimension of the associated characteristic space. This is done by solving some explicit equations associated to the slopes of the lower hull of the set of points $(p^i,\val(q_i))$ and by counting the multiplicity. When the system is regular singular at $0$, these eigenvalues are precisely the eigenvalues of the matrix $M_d$ whose associated eigenvectors belong to $\ker N_d$. Then, there are only a finite number of constant matrices in Jordan normal form having this precise set of eigenvalues. Thus, one could test if, for each one of these matrices, there is a basis of solutions in the Puiseux series by applying the cyclic vector lemma (Algorithm \ref{Algo:Cyclicbis}) and by solving $m$ equations of the form \eqref{eq:system_homogene_in_ramification}, where $A$ is a companion matrix, by extending the results of \cite[Algo. 7]{CDDM18} to the inhomogeneous case. By doing so, one could determine if a given Mahler system is regular singular at $0$. However, this method seems less efficient than the one presented in Algorithm \ref{Algo:dfixed}. Furthermore, by doing so, one would possibly have to work in finite extensions of the base number field ${\bf k}$ instead of the base field ${\bf k}$ (see Corollary \ref{coro:cdn}), in contrast to the method presented in this paper.

When the system \eqref{eq:Mahler_at_0} is regular singular at $0$, Algorithm \ref{Algo:dfixed} computes a $\K$-equivalent constant matrix. Furthermore, Roques \cite[§5.2]{Ro18} described fundamental matrices of solutions for constant systems. Precisely, for $c\in\Q^\star$ we let $e_c$ and $\ell$ denote functions such that $\puip\left(e_c\right)=ce_c$ and $\puip(\ell)=\ell+1$. Any constant system has a basis of solutions in $\Q\left[(e_c)_{c \in \Q^\star},\ell\right]$.

\begin{coro}
Consider a system \eqref{eq:Mahler_at_0} which is regular singular at $0$. From Algorithm \ref{Algo:dfixed}, one can compute a fundamental matrix of solutions of \eqref{eq:Mahler_at_0} with entries in $\K\left[(e_c)_{c \in \Q^\star},\ell\right]$.
\end{coro}

For example, one can take respectively for $e_c$ and $\ell$, the functions $\log(z)^{\log(c)/\log(p)}$ and  $\log\left(\log(z)\right)/\log(p)$. Before proving Proposition \ref{prop:complexity}, we make some observations about the shape of the matrices $M_d$ and $N_d$.

\subsection{On the shape of \texorpdfstring{$M_d$}{Md} and \texorpdfstring{$N_d$}{Nd}}

Algorithm \ref{Algo:dfixed} requires some calculations with the matrices $M_d$ and $N_d$. Naively, it should necessitate $\MM(n)$ operations where $n$ is at least the number of rows and the number of columns of $M_d$ and $N_d$. However, by looking more closely at the shape of $M_d$ and $N_d$, we will show that it can be lowered to $d\MM(n/d)$.

\begin{defi}
Let $D=\left(D_{i,j}\right)_{1\leq i\leq r, 1\leq j\leq s}$ be a block matrix with $D_{i,j}\in\mathcal{M}_m\left(\Q\right)$. We say that $D$ is a \emph{$d$-gridded matrix} if for all $(i_0,j_0)\in \lbrace 1,\ldots,r\rbrace\times\lbrace 1,\ldots,s\rbrace$ such that $D_{i_0,j_0}$ is nonzero, the matrices $D_{i_0,j}$, $D_{i,j_0}$ with $i\not\equiv i_0\mod(d)$ and $j\not\equiv j_0\mod(d)$ are zero matrices. Let $\sigma$ be a permutation of the set $\{1,\ldots,d\}$. We say that $\sigma$ is \emph{associated with} the $d$-gridded matrix $D$ if $D_{i,j}=0$ for every $(i,j) \in \{1,\ldots,r\}\times\{1,\ldots,s\}$ with $j \not\equiv \sigma(i)\mod(d)$.
\end{defi}

\begin{lem}\label{lem:Operationgridded}
Let $D=(D_{i,j})_{1\leq i\leq r, 1\leq j\leq s}$ and $E=(E_{i,j})_{1\leq i\leq s, 1\leq j\leq t}$ be two $d$-gridded matrices with $D_{i,j}, E_{i,j}\in\mathcal{M}_m\left(\Q\right)$ and, respectively, $\sigma_D$ and $\sigma_E$ their associated permutation. We write $u:=\max\left(r,s,t\right)$. The computation of the product $DE$ can be done with complexity
$$
\mathcal O(d\MM(mu/d)).
$$
Furthermore, $DE$ is a $d$-gridded matrix with associated permutation $\sigma_E\circ \sigma_D$.
\end{lem}

\begin{proof}
We let 
$D_{n}$ (respectively $E_{n}$) denote the block matrices $(D_{n+k d,\sigma_D(n)+\ell d})_{k,\ell}$ 
(\mbox{respectively} $(E_{n+k d,\sigma_E(n)+\ell d})_{k,\ell}$) for any $n \in \{1,\ldots,d\}$. 
Let $n_0 \in \{1,\ldots,d\}$, write $F_{n_0}:=D_{n_0} E_{\sigma_D(n_0)}$
and consider $F_{n_0}:= (F_{n_0,k,\ell})_{k,\ell}$ its block decomposition, where $F_{n_0,k,\ell}\in\mathcal{M}_m\left(\Q\right)$. 
For any $i\in\{1,\ldots,r\}$, write $i=n_0+kd$ with $n_0 \in \{1,\ldots,d\}$, $k\in\mathbb{N}$ and for any $j\in\{1,\ldots,t\}$, let
$$
G_{i,j}:=
\left\{ \begin{array}{cl}
F_{n_0,k,\ell} & \text{ if } j=\sigma_E\circ \sigma_D(n_0)+\ell d\, \text{ for some } \ell\in\mathbb{N} \, ,
\\ 
0 & \text{ otherwise}\,. 
\end{array} \right.
$$
Then $DE=(G_{i,j})_{i,j}$ and it is a $d$-gridded matrix whose associated permutation is $\sigma_E\circ \sigma_D$. The complexity of computing the product of two permutations of $\{1,\ldots,d\}$ is $\mathcal O(d)$. Then, the complexity of computing each matrix $F_n$ is $\mathcal O(\MM(mu/d))$. Thus, the complexity of computing $DE$ is
$$
\mathcal O(d+d\MM(mu/d))=\mathcal O(d\MM(mu/d))\,. 
$$
\end{proof}

\begin{rem}
The computation of a basis of the (right-)kernel of a $d$-gridded matrix can be done with the same complexity as the product of two $d$-gridded matrices. Note that we can add some zero column vectors to the column vectors of the kernel obtained in this way in order to form a $d$-gridded matrix. Similarly, one can compute a basis of the intersection of the vector spaces spanned by the columns of two $d$-gridded matrices with the same complexity. The basis obtained being a subset of the columns of one of the matrices, one can complete it with some zero column vectors in order to form a new $d$-gridded matrix.
\end{rem}

\begin{lem}\label{lem:M_N_gridded}
Let $d \in \mathcal D_0$. The matrices $M_d$ and $N_d$ are $d$-gridded matrices.
\end{lem}
\begin{proof}
Recall that 
$$
M_d:=(\Ainv_{d,i-pj})_{\nu_d \leq i,j\leq \mu_d} \text{ and } N_d:=(\Ainv_{d,i-pj})_{
\val\left(\Ainv\right)+p\nu_d \leq i \leq \nu_d-1,\, \nu_d \leq j \leq \mu_d}$$
where $\puid(A^{-1}):=\sum_{n} \Ainv_{d,n} z^n$. In particular, $\Ainv_{d,i-pj}=0$ if $d$ does not divide $i-pj$. Thus if $\Ainv_{d,i_0-pj_0}\neq 0$ then $\Ainv_{d,i-pj_0}=0$ for all $i$ such that $i\not\equiv i_0\  (\textrm{mod}\ d)$. 
Moreover, since $p$ and $d$ are relatively prime, if $\Ainv_{d,i_0-pj_0}\neq 0$ then $\Ainv_{d,i_0-pj}=0$ for all $j$ such that $j\not\equiv j_0\ (\textrm{mod}\ d)$. Associated permutations to these matrices are $\sigma_M$ and $\sigma_N$ such that, for every $k \in \{1,\ldots,d\}$, 
$$
p\sigma_M(k) \equiv (p-1)(1-\nu_d)+k\, (\textrm{mod}\ d) 
$$
$$
p\sigma_N(k) \equiv \val\left(\Ainv_d\right)+p-1+k\, (\textrm{mod}\ d)\, .
$$
\end{proof}

\subsection{Proof of Proposition \ref{prop:complexity}}\label{sec:proof-complexity}

We recall that 
$$
\intersd = \intersd^{+} \bigcap \intersd^{-}
$$
where 
$$
\intersd^{+} := \bigcap_{n \in \N} M_d^n \ker(N_d)\,,\quad \intersd^{-}:= \bigcap_{n \in \N} \ker (N_dM_d^n).
$$
We first use the two following lemmas to prove that the vector space $\inters$ in Algorithm \ref{Algo:dfixed} is equal to the vector space $\intersd$.

\begin{lem}\label{lem:rangeX}
Let $c_d := m\left(\mu_d-\nu_d+1\right)$. For any $c \geq c_d$, the vector space $\intersd$ is the image of 
$\bigcap_{n=0}^{c-1} \ker(N_dM_d^n)$
under the left multiplication by $M_d^{c}$.
\end{lem}

\begin{proof}
We first prove that $\intersd^- = \bigcap_{n=0}^{c-1} \ker(N_dM_d^n)$. Write $\V_n=\cap_{k=0}^{n-1} \ker\left( N_dM_d^k\right)$. It is clear that if $\V_n=\V_{n+1}$ then $\V_\ell=\V_n$ for all $\ell\geq n$. Thus the sequence $(\V_n)_{n \geq 1}$ is decreasing and then stationary. Since $\dim \V_1 \leq c_d \leq c$, we must have $\V_{c} = \V_{c+1}$ and $\intersd^-=\lim_{n\to \infty} \V_n=\V_{c}$.

Now, write $\W_n= M_d^n\intersd^{-}$. Since $M_d\intersd^-\subset \intersd^-$, the sequence $(\W_n)_{n\in \N}$ is non-increasing. We prove that 
$\intersd = \lim_{n\to \infty } \W_n$.
Since $\intersd^{-} \subset \ker(N_d)$ we have
$$
\lim_{n\to \infty } \W_n= \bigcap_{n \in \N} M_d^n\intersd^{-}  \subset  \bigcap_{n \in \N} M_d^n\ker(N_d)  \cap \intersd^{-} \subset \intersd\,.
$$
It remains to prove that $\intersd \subset \W_n$ for every $n$. We argue by induction on $n$. When $n=0$ it is immediate since $\W_0=\intersd^{-}$. Assume now that $\intersd\subset \W_n$ for some $n\geq 0$. Let $x\in\intersd$. It follows from Lemma \ref{lem:X_max} that $M_d\intersd = \intersd$. Hence, there exists $y\in\intersd$ such $x=M_dy$. By assumption, $y\in \W_n=M_d^n\intersd^-$ so $x\in M_d^{n+1}\intersd^- =\W_{n+1}  $, which concludes the induction. Now, arguing as in the first part of the proof, $(\W_n)_{n \in \N}$ is stationary after the rank $c_d$. In particular, since $c\geq c_d$, $\intersd=\lim_{n\to \infty }\W_n = \W_{c}$.

\end{proof}

Let $t$ be the least integer such that $2^t \geq m(\mu_d-\nu_d+1)$. We now define recursively a finite sequence of vector spaces $(\I_\ell)_{0\leq \ell \leq t}$ setting
$$
\I_0:= \ker(N_d) \text{ and } \I_\ell:=\lbrace \x\in\I_{\ell-1}\mid M_d^{2^{t-\ell}}\x\in \I_{\ell-1}\rbrace
$$
\begin{lem}\label{lem:constructionX}
We have
$$
\inters_d=M_d^{2^t}\I_t\,.
$$
\end{lem}
\begin{proof}
One checks by induction on $\ell\in\lbrace 0,\ldots, t\rbrace$ that
$$
\I_\ell =\bigcap_{n=0}^{2^\ell-1} \ker\left(N_d M_d^{n2^{t-\ell}}\right)\,.
$$ 
Thus, $\I_t=\bigcap_{n=0}^{2^t-1} \ker(N_d M_d^n)$ and the result follows from  Lemma \ref{lem:rangeX}.
\end{proof}

\begin{proof}[Proof of Proposition \ref{prop:complexity}] Let $d \in \mathcal D_0$ be given by Algorithm \ref{Algo:dbis}. We infer from Lemma \ref{lem:constructionX} that the vector space $\inters$ in Algorithm \ref{Algo:dfixed} is equal to $\intersd$. Thus, from Theorem \ref{thm:Espace}, Algorithm \ref{Algo:dfixed} returns ``true'' if and only if the system is regular singular at $0$. Then, arguing as in the proof of Theorem \ref{thm:Espace}, one checks that $A$ is $\mathbf K$-equivalent to $R^{-1}$ and that $\sum_{j=\nu_d}^{dn} E_j z^{j/d}$ are the first coefficients in the Puiseux expansion of an associated gauge transformation.

To compute the complexity, we follow the script of Algorithm \ref{Algo:dfixed}. Let $\delta:=\deg(A)$. From Proposition \ref{prop:ComplexityCyclic} Algorithm \ref{Algo:dbis} computes the integer $d$ with
\begin{equation}\label{eq:complexity_Algo2}
\mathcal O\left(\MM(m)\log(m)\M\left(u\right)\log(u) \right)
\end{equation}
operations, where $u:= (m +\delta)p^{m}$.
To compute $M_d$ and $N_d$ one needs to compute the Laurent series expansion of $A^{-1}$ between 
$\val\left(A^{-1}\right)$ and $(\mu_d-p\nu_d)/d$. The computation of the inverse of $A$ can be done with complexity
\begin{equation}
\label{eq:compute_inv_A}
\mathcal O(m^2\M(m\delta)\log(m\delta))
\end{equation}
Newton's method allows to compute the $n$ first terms in the Laurent series expansion of a rational function of degree at most $n$ with complexity $\mathcal O(\M(n))$. One checks that $\deg\left(A^{-1}\right)\leq m\delta$. Let $v:=-(\val(A)+\val(A^{-1}))+1\geq 1$. One has
$$
 n_0:=\frac{\mu_d-p\nu_d}{d}-\val\left(A^{-1}\right)=\mathcal O\left(v\right)  
$$ 
and $v\leq m\delta$. Thus the complexity of computing the first $n_0$ terms of the Laurent expansion of the $m^2$ entries of $A^{-1}$ is $\mathcal{O}\left(m^2\M\left(m\delta\right)\right)$, which is negligible compared to \eqref{eq:compute_inv_A}. 
Thus, the computation of $M_d$ and $N_d$ can be done with complexity \eqref{eq:compute_inv_A}.
Let $t$ be the least integer such that $2^t>m(\mu_d-\nu_d+1)$. The cost of computing $t$ is negligible.
We compute $M_d,M_d^2,\ldots,M_d^{2^{t-1}}$. The number of rows and columns of $M_d$ being $\mathcal O(mdv/p)$, it follows from Lemma \ref{lem:Operationgridded} that is necessitates $\mathcal O(t d\MM(mv/p))$ operations. 
We compute $\I_0,\ldots,\I_t$.
Since $N_d$ has $\mathcal O(mdv)$ rows and columns, the complexity of computing $\I_0$ is $\mathcal O(d\MM(mv))$. Knowing $\I_{\ell-1}$, the complexity of computing $\I_\ell$ is $\mathcal O(d\MM(mv/p))$. Thus, the complexity of computing the whole sequence is
$$
\mathcal{O}\left(td\MM(mv/p)+d\MM(mv)\right)\,.
$$
We now compute $\inters$. Since we know $M_d^{2^{t-1}}$, the complexity of computing $M_d^{2^t}\I_t$ is $\mathcal O(d\MM(mv/p))$.
Since $d \leq p^m$ and since $t=\mathcal O(\log(mp^{m-1}v))$, the complexity of the computation of $\inters$ is
\begin{equation}\label{eq:complexite_inters}
\mathcal O\left(p^m\MM(mv/p)\log(mp^{m-1}v) + p^m\MM(mv)\right)\,.
\end{equation}
Now, \eqref{eq:compute_inv_A} is negligible with respect to \eqref{eq:complexity_Algo2}. Thus, Algorithm \ref{Algo:dfixed} returns if a system is regular singular or not with
$$
\mathcal O\left(\MM(m)\log(m)\M(u)\log(u)+p^m\MM(mv/p)\log(mp^{m-1}v) + p^m\MM(mv)\right)
$$
operations. Using the notation $\widetilde{\mathcal O}$ and, remembering that $\log(p^m)=m\log(p)$, we obtain the bound we want.
\end{proof}

\begin{rem}
In Algorithm \ref{Algo:dfixed}, we chose to compute first the integer $d$ thanks to the cyclic vector lemma, Algorithm \ref{Algo:Cyclicbis} and Algorithm \ref{Algo:dbis}. Then we computed the vector space $\intersd$ with this $d$. One could ask if running the algorithm for every $d \in \mathcal D$ could be faster. It does not seem to be the case. Since we only have to compute the inverse of $A$ once and since $\mathcal D$ has $\mathcal O(p^{m})$ elements, running the algorithm for every $d \in \mathcal D$, without using Algorithm \ref{Algo:dbis}, would necessitate
$$
\widetilde{ \mathcal O}(m^2\M(m\deg(A))+ mp^{2m}\MM(mv))
$$
operations. When $\deg(A)$ is large compared to other parameters, it can be smaller than the complexity of Algorithm \ref{Algo:dfixed}. However, we have to pay a factor $p^{2m}$ instead of $p^{m}$.
\end{rem}

\section{Examples}\label{sec:Example}

In this section, we study the regular singular property of some particular systems.

\subsection{Systems of size 1}

We consider a system of size $1$:

\begin{equation}\label{eq:eq_order_1}
\puip(y)=ay    
\end{equation}
where $a \in \Q(z)$, $a\neq 0$.
\begin{prop}
\label{prop:homogen_order1_RS}
Any system of size $1$ is regular singular at $0$.
\end{prop}

\begin{proof}
We consider the equation \eqref{eq:eq_order_1}. Let $\nu$ denote the valuation at $0$ of $a$ and define $\psi:=z^{\nu/(p-1)}$. Then, the system $\puip(y)=by$ with $b:=\puip\left(\psi\right)^{-1}a\psi$ is strictly Fuchsian at $0$. Thus, the homogeneous equation \eqref{eq:eq_order_1} is $\Q\left(\left(z^{\nu/(p-1)}\right)\right)$-equivalent to an equation which is strictly Fuchsian at $0$. \textit{A fortiori}, \eqref{eq:eq_order_1} is regular singular at $0$.
\end{proof}

\subsection{An equation of order 2}\label{subsec:exempleCDDM}
Consider the $3$-Mahler equation:
\begin{multline}
\nonumber
z^3(1 - z^3 + z^6 )(1 - z^7 - z^{10}) \phi_3^2(y)
- (1 - z^{28} - z^{31} - z^{37} - z^{40}) \phi_3(y)
\\ + z^6 (1 + z)(1  - z^{21} - z^{30} )y = 0\, .
\end{multline}
The matrix of the $3$-Mahler system associated with this equation is
$$
A(z):=\begin{pmatrix}
0& 1
\\
- \frac{ z^3 (1 + z)(1  - z^{21} - z^{30} )}{(1 - z^3 + z^6 )(1 - z^7 - z^{10}) } & 
\frac{1 - z^{28} - z^{31} - z^{37} - z^{40}}{z^3(1 - z^3 + z^6 )(1 - z^7 - z^{10}) } 
\end{pmatrix}\, .
$$
We propose to check whether or not the $3$-Mahler system associated with this matrix is regular singular at $0$. Since we already know a homogeneous linear equation associated with this system, it is not necessary to run Algorithm \ref{Algo:Cyclicbis}. Algorithm \ref{Algo:dbis} applied to this system returns $d:=2$. We now run Algorithm \ref{Algo:dfixed} with $d=2$. We have $\val(A)=-3$, $\val(A^{-1})=-6$ and thus $
\nu_2=-3$  and $\mu_2= 6$. In that case, the vector space $\inters_2$ is spanned by the transpose of the two linearly independent vectors
$$
\begin{array}{c}
(0, 1,0,0, 1 ,0,0,0, -1 ,0,0,0, 1 , -1 ,0,0, -1,0,0,0)\, ,
\\ 
(0,0,0,0,0,0,0,0,0,0,0,0,0,0,0,0,0,0, -1,0)\,.
\end{array}
$$
In particular, $\inters_2$ has dimension $2$ and, from Theorem \ref{thm:Espace}, the system is regular singular at $0$. One can check that these vectors are eigenvectors of the matrix $M_2$ for the eigenvalue $1$. Thus the matrix $R$ is the identity matrix of size $2$. In particular, the associated gauge transformation $\Gtr$ given by Algorithm \ref{Algo:dfixed} is a fundamental matrix of solutions because it satisfies 
$$
\phi_3\left(\Gtr\right)^{-1}A\Psi = \id_2.
$$
From these two vectors, we can compute the first terms of the Puiseux expansion of $\Gtr$
$$
\Gtr = \begin{pmatrix}
f_1 & f_2 \\
f_3 & 0
\end{pmatrix} + \mathcal O(z^{17/2})
$$
with 
\begin{eqnarray*}
f_1(z)& =& z^{-1/2}-z^{1/2}+z^{3/2}-z^{5/2}+z^{7/2}-z^{9/2}+z^{11/2}-z^{13/2}+z^{15/2}\,,\\ f_2(z) &=& -z^3+z^4-z^5+2z^6-2z^7+2z^8\,,
\\ f_3(z) &=& z^{-3/2}-z^{3/2}+z^{9/2}-z^{15/2}\,.
\end{eqnarray*}

\begin{rem}
Note that this example is the same as the one that the authors of \cite{CDDM18} chose to illustrate their paper.  
\end{rem}

\subsection{Systems coming from finite deterministic automata}

As mentioned in the introduction, Mahler systems are related to the automata theory. Indeed, the generating function of an automatic sequence (see \cite{AS03} for a definition) is solution of a Mahler equation. Numerous famous automatic sequences are related to homogeneous or inhomogeneous Mahler equations of order $1$. This is, for example, the case of the \textit{Thue-Morse sequence}, the \textit{regular paper-folding sequence}, the sequences of powers of a given integer, the characteristic sequence of \textit{triadic Cantor integers} -- those whose base-$3$ representation contains no $1$. Thus, their associated systems are regular singular at $0$. 

Among the sequences satisfying equations with an order greater than $1$, a famous one is the \textit{Baum-Sweet sequence}, the characteristic sequence of integers whose binary expansion has no blocks of consecutive $0$ of odd length. The system associated with this sequence is strictly Fuchsian at $0$ and thus regular singular at $0$. Another important one is the Rudin-Shapiro sequence whose general term is
$$
\left\{\begin{array}{cccl} a_n & = & 1 & \text{ if the number of occurrences of two consecutive } 1 
\\ &&&\text{ in the binary expansion of } n \text{ is even}
\\ a_n&=&-1& \text{ otherwise}.
\end{array}\right.
$$ 
Its generating series $f:=\sum_{n\in \N} a_n z^n$ satisfies the equation
$$
\phi_2\left(\begin{array}{c} f(z) \\ f(-z)\end{array}\right)
=\frac{1}{2}
\left(\begin{array}{cc} 1 & 1 \\ \frac{1}{z}& \frac{-1}{z}  \end{array}\right)\left(\begin{array}{c} f(z) \\ f(-z)\end{array}\right)\,.
$$
This system is not regular singular. Indeed, Algorithm \ref{Algo:dbis} returns $d=3$ and we have $\dim \inters_3=1$. Thus, Algorithm \ref{Algo:dfixed} returns `` False ''.
 
The regular singular property can be seen as ``normal'' for Mahler systems since a sufficient condition is to be strictly Fuchsian at $0$. However, the generating series of an automatic sequence satisfies a Mahler system with a very precise shape: $A^{-1}(0)$ is well defined and has at most one nonzero entry in each column. Among these systems, the strictly Fuchsian property is more occasional.

\section{Open problems}\label{sec:OpenQuestions}

We discuss here some open problems about the regular singularity at $0$ of a Mahler system.

\subsection{The inverse matrix system}

Let $A \in {\rm GL}_m(\Q(z))$ and $p \geq 2$ be an integer. If the $p$-Mahler system with matrix $A$ is strictly Fuchsian at $0$, then the $p$-Mahler system with matrix $A^{-1}$ is also strictly Fuchsian at $0$ (and hence, regular singular at $0$). This property does not extend to regular singular systems. For example, if $A$ denotes the matrix of the regular singular system in subsection \ref{subsec:exempleCDDM}, the $3$-Mahler system associated with $A^{-1}$ is not regular singular at $0$. We ask the following question.

\noindent \emph{Is there a characterisation of matrices $A$ such that the $p$-Mahler systems associated with both $A$ and $A^{-1}$ are regular singular at $0$?}

\subsection{Changing the Mahler operator}

Assume that a system is strictly Fuchsian at $0$. If we change the integer $p$ then the system remains strictly Fuchsian at $0$ (hence regular singular at $0$). This property does not extend to regular singular systems. Indeed, the $3$-Mahler system of subsection \ref{subsec:exempleCDDM} is regular singular at $0$, while the $2$-Mahler system with the same matrix is not. Similarly, the $p$-Mahler system associated with this matrix is not regular singular when $p \in \{4,\ldots,30\}$ (and probably beyond). Similarly, the companion system associated with the $p$-Mahler equation
$$
(z^{11} + z^{13}) \puip^2(y) + (-1/z - z -z^6+z^{10}) \puip(y) + (1-z)y = 0
$$
is regular singular at $0$ for $p=2$ and $p=4$ but not for $p\in\{3,5,6,\ldots,100\}$ (and probably beyond). It seems that for a matrix $A \in {\rm GL}_m\left(\Q(z)\right)$ the $p$-Mahler system associated with $A$ is either regular singular at $0$ for every integer or for finitely many (possibly none) integers $p \geq 2$.

\noindent \emph{Is that true that only these two situations may occur?}

\bigskip

\noindent{\bf \itshape Acknowledgement.}\,\,--- The authors would like to thank Julien Roques for the valued discussions and his lights on his paper \cite{Ro20}, Thomas Dreyfus for his insights about the paper \cite{CDDM18} and Boris Adamczewski for his feedback on this work. They are grateful to the referees for their valuable advice on first versions of this paper.

\end{document}